\RequirePackage{amsmath}
\documentclass[11pt]{llncs}

\usepackage{amssymb}
\usepackage{amsfonts}
\usepackage{dsfont}
\usepackage[colorlinks=true, citecolor=blue]{hyperref}
\usepackage{color,graphicx,setspace}
\usepackage{geometry}
\usepackage{paralist}
\usepackage{todonotes}

\usepackage{thmtools}
\usepackage{thm-restate}
\usepackage{cleveref}

\usepackage{authblk}

\title{Faster Minimization of Total Weighted Completion Time on Parallel Machines}

\author{Danny Hermelin\inst{1} \and Tomohiro Koana\inst{2} \and Dvir Shabtay}

\institute{
Department of Industrial Engineering and Management, Ben-Gurion University of the Negev, Beer-Sheva
\email{hermelin@bgu.ac.il}
\and
Utrecht University, Netherlands \& Research Institute for Mathematical Sciences, Kyoto University, Japan, \email{tomohiro.koana@gmail.com}
\and
Department of Industrial Engineering and Management, Ben-Gurion University of the Negev, Beer-Sheva
\email{dvirs@bgu.ac.il}
}

\begin{document}

\sloppy
\maketitle

\begin{abstract}
We study the classical problem of minimizing the total weighted completion time on a fixed set of $m$ identical machines working in parallel, the $Pm||\sum w_jC_j$ problem in the standard three field notation for scheduling problems. This problem is well known to be NP-hard, but only in the ordinary sense, and appears as one of the fundamental problems in any scheduling textbook. In particular, the problem served as a proof of concept for applying pseudo-polynomial time algorithms and approximation schemes to scheduling problems. The fastest known pseudo-polynomial time algorithm for $Pm||\sum w_jC_j$ is the famous Lawler and Moore algorithm from the late 1960's which runs in $\tilde{O}(P^{m-1}n)$ time\footnote{Throughout the paper we use $\tilde{O}(·)$ to suppress logarithmic factors}, where $P$ is the total processing time of all jobs in the input. After more than 50 years, we are the first to present an algorithm, alternative to that of Lawler and Moore, which is faster for certain range of the problem parameters (\emph{e.g.}, when their values are all~$O(1)$). 
\end{abstract}


\section{Introduction}
\label{sec:intro}%


In manufacturing, \emph{holding cost} refers to the total expense associated with storing inventory. This cost typically includes factors such as employee wages, maintenance and upkeep, energy usage, utilities, and other related expenses. Additionally, holding costs may incorporate the opportunity cost of capital tied up in inventory, as well as the risk of inventory obsolescence or damage. Effectively managing holding costs is essential, as excessive expenses can erode profit margins, while insufficient storage investment can lead to operational inefficiencies. In many industrial scenarios, holding cost is a critical factor in determining a business's profitability and overall competitiveness.

In the simplest and most basic case, the holding cost is uniform throughout the entire manufacturing process. Under this assumption, we can associate a weight (or penalty)~$w_j$ with each job~$j$, representing the job holding cost per unit of time. Consequently, the total holding cost for job $j$ is given by~$w_j C_j$, where $C_j$ is the total time the job spends in the manufacturing process. Minimizing the total weighted completion time, expressed as~$\sum w_j C_j$, becomes a central objective in scheduling theory. This metric captures the trade-off between processing jobs quickly and the penalties incurred from delays. As such, it serves as a foundation for analyzing and optimizing scheduling policies in various industrial and computational settings.

Indeed, minimizing~$\sum w_jC_j$ is one of the most classical objectives in scheduling literature~\cite{Pinedo2008}. In fact, one of the seminal results in this field pertains to its single-machine counterpart, the $1||\sum w_j C_j$ problem, which can be solved in  $O(n \log n)$ time by applying the so-called WSPT (Weighted Shortest Processing Time) rule, also known as Smith's rule~\cite{Smith1956} (see Lemma~\ref{lem:smith}). Since then, this rule has served as a foundational building block for algorithmic results in scheduling problems, see \emph{e.g.}~\cite{HohnJacobs2015,Nicole,DBLP:conf/stoc/SkutellaW99,Weiss92}. In this paper we are interested in the generalization of this problem to the case where we have $m$ identical machines operating in parallel, for some fixed constant $m$. The so called $Pm||\sum w_j C_j$ problem. 

Formally, in the $Pm||\sum w_j C_j$ problem we are given a set of $n$ jobs, $J=\{1,\ldots,n\}$, which are available for non-preemptive processing on a set of $m$ identical machines operating in parallel. Each job $j$ is associated with two positive integer parameters: ($i$) its \emph{processing time}~$p_j$ on any of the $m$ machines, and ($ii$) its \emph{weight}~$w_j$ indicating its holding cost per unit of time. A \emph{schedule} (solution) consists of $m$ permutations $\sigma=(\sigma_1,\ldots,\sigma_m)$ such that for all $i \neq j \in \{1,\ldots,m\}$, the set of elements in $\sigma_i$ and the set of elements in $\sigma_j$ have no elements in common, and the union of these sets is $J$. We say that job $j$ is scheduled on machine $i$ if $j$ appears in $\sigma_i$. In this way, permutation $\sigma_i$ represents the \emph{processing order} of jobs processed by machine $i$. The \emph{completion time} $C_j$ of job $j$, assuming that it is scheduled on machine $i$, is simply the total processing time of all jobs processed before $j$ on machine $i$ (including the processing time $p_j$ of $j$ itself). Therefore, it can be computed by $C_j = \sum_{\sigma^{-1}_i(k) \leq \sigma^{-1}_i(j)} p_k$. The objective is to find a schedule that minimizes the total weighted completion time~$\sum_j w_jC_j$, which corresponds to the total holding cost of all jobs. 

The $Pm||\sum w_j C_j$ problem is NP-hard already for the two machine case, \emph{i.e.} the $P2||\sum w_j C_j$ problem. Indeed, when the weight of each job equals its processing time, the corresponding $P2||\sum p_j C_j$ problem is essentially the \textsc{Partition} problem~\cite{BrunoCoffman}. However, the problem is only weakly NP-hard, meaning it admits pseudo-polynomial algorithms. Indeed, Lawler and Moore~\cite{LawlerMoore} presented such an algorithm based on dynamic programming in the late 60s. Let $T_j[P']$ denote the minimum total weighted completion time of jobs $\{1,\ldots,j\}$ on two identical parallel machines, with the additional constraint that the total processing time of the jobs on the first machine is at most $P'$. Then, assuming the input set is ordered according to the WSPT rule, the Lawler and Moore algorithm computes the values~$T_j[P']$ in dynamic programming fashion via the following recursive function:
$$
T_j[P'] = \min
\begin{cases}
\quad T_{j-1}[P' - p_j] + w_j \cdot P',\\ 
\quad T_{j-1}[P'] + w_j \cdot (p_1 + \cdots + p_j -  P'), \\
\end{cases}
$$
where $T_0[P']=0$ for all possible values of $P'$. Using this recursive function, the Lawler and Moore algorithm can compute the minimum total weighted completion time of any $P2||\sum w_jC_j$ instance in $O(P \cdot n)$ time, where $P= \sum_j p_j$. 

The dynamic program described above can be naturally extended to the case of a constant $m > 2$  number of machines by computing tables $T_j[P'_1,\ldots,P'_{m-1}]$, where the entries correspond to the states of the first $m-1$ machines. This approach results in an algorithm with a running time of $O(P^{m-1} \cdot n)$ for the multiple machine case. Remarkably, even though this algorithm was devised over fifty years ago, it remains the fastest known pseudo-polynomial time algorithm for solving the $Pm||\sum w_jC_j$ problem. This lack of improvement is particularly surprising given the simplicity and fundamental nature of the algorithm. Despite its basic structure, another more efficient approach has yet been discovered, highlighting an intriguing gap in our understanding of this problem's computational complexity. This observation serves as the starting point of this paper.

\subsection{Our results}
The main contribution of this paper is a new pseudo-polynomial time algorithm for the $Pm||\sum w_j C_j$ problem, improving upon Lawler and Moore’s algorithms for a specific range of the problem parameters. Let $p_{\max} = \max_j p_j$ and $w_{\max} = \max_j w_j$ respectively denote the maximum processing time and weight in a given $Pm||\sum w_j C_j$ instance $J$, and let $P= \sum_j p_j$. Our main result is stated in the following theorem.
\begin{restatable}{theorem}{main}
\label{thm:main}%
$Pm||\sum w_jC_j$ can be solved in $\widetilde{O}(P^{m-1} + w_{\max}^{m+1} \cdot p_{\max}^{4m+1})$ time.
\end{restatable}
\noindent Thus, when $w_{\max}$ and $p_{\max}$ are sufficiently small, this running time is indeed faster than Lawler and Moore's $O(P^{m-1} \cdot n)$ time algorithm. For example, when both $w_{\max}$ and $p_{\max}$ are bounded by some constant, our algorithm runs in $\widetilde{O}(n^{m-1})$ time, while Lawler and Moore's algorithm runs in~$O(n^m)$ time. Specifically, our algorithm is the first (nearly) linear time algorithm for $P2||\sum w_jC_j$ with $w_{\max},p_{\max}=O(1)$.

\begin{figure}[ht!]
\centering
\includegraphics[width=110mm]{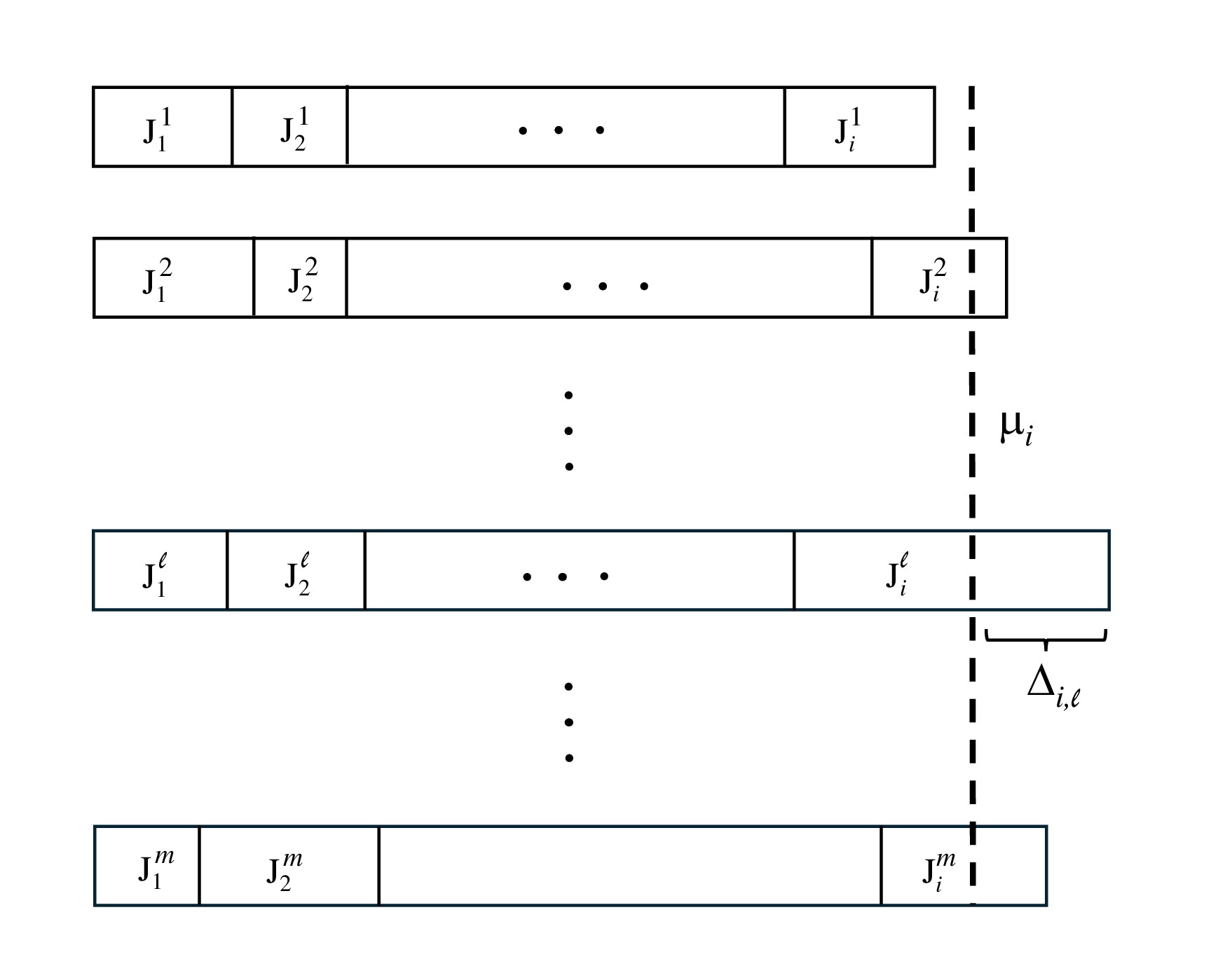}
\caption{A depiction of a WSPT schedule $\sigma$ for jobs with efficiency at least $e_{(i)}$, along with the average processing time $\mu_i$ of these jobs, and the deviation of $\sigma$ for jobs with efficiency at least $e_{(i)}$ on machine $\ell$ (which is negative in the example). 
\label{fig:wjCj}}%
\end{figure}

A central component of our algorithm is a new theorem concerning the structure of optimal schedules for the $Pm||\sum w_jC_j$ problem. Define the \emph{efficiency} $e_j$ of job $j$ as~$e_j = w_j / p_j$. According to Smith's rule, there always exists an optimal schedule for any $Pm||\sum w_jC_j$ instance $J$ where the jobs are scheduled in non-increasing order of efficiencies on each machine. We refer to such schedules as WSPT (Weighted Shortest Processing Time) schedules. Let $e_{(1)} > \cdots > e_{(k)}$ denote the \emph{distinct} efficiencies among the jobs in the input set, and define $\mu_i$ to be the \emph{average processing time (or load)} of jobs with efficiency at least $e_{(i)}$; that is, the total processing time of all such jobs divided by the number of machines~$m$. For a given WSPT schedule $\sigma$, and for  $1 \leq i \leq k$ and $1 \leq \ell \leq m$, let~$\Delta_{i,\ell}$ be the \emph{deviation} of $\sigma$ for jobs with efficiency at least $e_{(i)}$ on machine $\ell$; \emph{i.e.}, the difference between $\mu_i$ and the total processing time of such jobs that $\sigma$ schedules on machine $\ell$ (see Figure~\ref{fig:wjCj}). 
\begin{restatable}{theorem}{structure}
\label{thm:structure}%
In any optimal WSPT schedule for any $Pm||\sum w_j C_j$ instance with $k$ distinct efficiencies, we have 
$$
|\Delta_{i,\ell}| \leq  \sqrt{m \cdot w_{\max}} \cdot p_{\max}^2
$$ 
for each $1 \leq i \leq k$ and~$1 \leq \ell \leq m$.    
\end{restatable}

We also demonstrate that our algorithm can be further improved for the two-machine case by utilizing a slightly different structure theorem. Specifically, let $\delta_i$ denote the difference between the load of jobs with efficiency at least $e_{(i)}$ on the first and second machine. Similar to the above, we can show that $|\delta_i| \leq \sqrt{w_{\max}} \cdot p^2_{\max}$ for any efficiency~$e_{(i)}$. However, as it turns out, using the $\delta_i$ values instead of the the $\Delta_{i,\ell}$ values leads to an improved running time for the two machine case. Our secondary result is stated in the following theorem.
\begin{restatable}{theorem}{secondary}
\label{thm:secondary}%
$P2||\sum w_jC_j$ can be solved in $\widetilde{O}(P + w_{\max}^{2} \cdot p_{\max}^{5})$ time.
\end{restatable}


\subsection{Related work}

Several results concerning the objective of minimizing the total weighted completion time on parallel machines serve as key cornerstones in the field of scheduling theory. As mentioned above, one of the seminal results in the field is the $O(n \log n)$ time algorithm for the single machine $1||\sum w_jC_j$~\cite{Smith1956}. Another classical variant of the problem that can be solved in $O(n \log n)$ time is $P|| \sum C_j$ where jobs have unit-weight and the scheduling is done on an arbitrary number of identical parallel machines~\cite{ConwayMM1967} ($P$ replaces $Pm$ in the first field of the three field notation to indicate that the number of machines is not fixed). Nevertheless, other special cases are much harder to solve. For example, when $p_j=w_j$ for $j\in\{1,\ldots,n\}$ the corresponding $Pm||\sum p_j C_j$ is NP-hard for any fixed $m \geq 2$~\cite{BrunoCoffman}. Moreover, when $m$ is arbitrary, the corresponding $P||\sum w_j C_j$ becomes strongly NP-hard (see problem SS13 in Garey and Johnson~\cite{GareyJohnson}). The fact that $Pm||\sum w_j C_j$ is only weakly NP-hard is due to Lawler and Moore's pseudo-polynomial time algorithm mentioned above~\cite{LawlerMoore}. Jansen and Kahler~\cite{Jansenee} showed that $P2||\sum w_j C_j$ cannot be solved in $\tilde{O}(P^{1-\varepsilon})$ time for any $\varepsilon > 0$ assuming SETH.

A few variants of $Pm||\sum w_jC_j$ have been studied from the perspective of parameterized complexity. Mnich and Wiese~\cite{MnichW15} studied a variant where rejection is allowed, and the accepted jobs are scheduled on a single machine. They presented FPT and W[1]-hardness results regarding various parameters of the problem.
Knop and Kouteck{\'{y}}~\cite{DBLP:journals/scheduling/KnopK18} showed that $P||\sum w_j C_j$ is fixed parametrized tractable (FPT) for parameter $m+k$, where $k$ is the number of distinct efficiencies in the input set of jobs. This is done by formulating the problem as an $n$ fold integer programming, and using Hemmecke \emph{et al}~\cite{DBLP:journals/mp/HemmeckeOR13} algorithm to solve the later problem. Complimenting this result, Knop and Kouteck{\'{y}} also showed that the problem is W[1]-hard when only $m$ is taken as a parameter.

Despite limited progress in exact algorithms for minimizing total weighted completion times on parallel machines, substantial advances have been made in approximation methods. Sahni~\cite{SahniFPTAS} used the $P2||\sum w_jC_j$ problem as a key example to demonstrate the construction of fully polynomial-time approximation schemes (FPTAS) for scheduling problems. He presented a~$(1+\varepsilon)$-approximation algorithm for the problem, which runs in $O(n^2/\varepsilon)$ time. A classical result regarding the $P||\sum w_j C_j$ problem (the variant where the number of machines is arbitrary) shows that applying the WSPT rule greedily leads to a $(1+\sqrt{2})/2$-approximation~\cite{KK1986}. Skutella and Woeginger~\cite{DBLP:conf/stoc/SkutellaW99} improved this result by providing a polynomial time approximation scheme (PTAS) for the $P||\sum w_j C_j$ problem. A PTAS for the more general case where
machines are unrelated and jobs have arbitrary release dates appears in~\cite{DBLP:conf/focs/AfratiBCKKKMQSSS99}. Over the years, various approximation algorithms have been developed for several other strongly NP-hard generalizations of $1||\sum w_j C_j$ focusing on problems with arbitrary release dates and/or with precedence constraints (see for instance~\cite{DBLP:conf/stoc/BansalSS16,DBLP:journals/mp/CorreaM22,DBLP:conf/focs/ImL16,DBLP:conf/sosa/JagerW24,Jaeger:18:Approximating-total-weighted,DBLP:conf/focs/Li17,SittersYang:17:A-2+epsilon-approximation}).

\section{Preliminaries}

We let $J=\{1,\ldots,n\}$ denote a given $Pm||\sum w_j C_j$ instance, where the \emph{processing time} and \emph{weight} of job $j \in J$ are respectively denoted by $p_j$ and $w_j$. The \emph{efficiency} of a job~$j$ is defined as~$e_j = w_j / p_j$. The following cornerstone structural result for $Pm||\sum w_j C_j$, known as Smith's rule, states that it is always optimal to schedule jobs on each machine in non-increasing order of efficiencies.  
\begin{lemma}[\textbf{Smith's rule}~\cite{Smith1956}]
\label{lem:smith}
There exists an optimal schedule for $Pm||\sum w_jC_j$ where the jobs are scheduled in non-increasing order of efficiencies (WSPT order) on each machine.
\end{lemma}

We refer to (not necessarily optimal) schedules as in Lemma~\ref{lem:smith} above as \emph{WSPT schedules}. Let $e_{(1)} > \cdots > e_{(k)} > e_{(k+1)} = 0$ denote the \emph{distinct} efficiencies of jobs in~$J$, and let $J_i =\{j \in J: e_j = e_{(i)}\}$ denote the subset of jobs with efficiency~$e_{(i)}$, for each $i \in \{1,\ldots,k\}$. For a given WSPT schedule $\sigma$, we let $J^\ell_i(\sigma)$ denote the set of jobs with efficiency~$e_{(i)}$ that are scheduled by $\sigma$ on machine $\ell \in \{1,\ldots,m\}$. We will often omit $\sigma$ from this notation when the WSPT schedule in question is clear from the context. 
For a subset of jobs $J' \subseteq J$, we write~$P(J')= \sum_{j \in J'} p_j$ (and so~$P=P(J)$). 
For each $i \in \{ 1, \dots, k \}$, we define $\mu_i$ as $\frac{1}{m} P(J_1 \cup \cdots \cup J_i)$.

Goemans and Williamson~\cite{GoemansWillaimson00} showed an alternative, often more convenient way for computing the total weighted completion time of WSPT schedules. 
\begin{lemma}[\cite{GoemansWillaimson00}] 
\label{lem:GoemansWilliamson}%
The total weighted completion time of any WSPT schedule $\sigma$ for $J$ is
\begin{align*}
f(\sigma) = \frac{1}{2} \sum_{i = 1}^{k}  \Big[ \big(e_{(i)}-e_{(i+1)}\big)  \cdot \sum^m_{\ell=1} P(J^\ell_1 \cup \cdots \cup J^\ell_i)^2 \Big] + \frac{1}{2} \sum^n_{j =1} w_j p_j.
\end{align*}
\end{lemma}

\section{Structure of optimal WSPT schedules:}

In the following section we prove Theorem~\ref{thm:structure}, the main structural result of the paper restated below: 
\structure*
\noindent Recall that $\Delta_{i,\ell}$ is the difference between the average load of jobs with efficiency at least~$e_{(i)}$ and the load of these jobs on machine $\ell$, \emph{i.e} $\Delta_{i,\ell}=\mu_i -P(J^\ell_1 \cup \cdots \cup J^\ell_i)$ for each $1 \leq i \leq k$ and $1 \leq \ell \leq m$ (see Figure~\ref{fig:wjCj}). We prove Theorem~\ref{thm:structure} through a series of claims. 



\subsection{Yet another alternative formulation}

Our first step towards proving Theorem~\ref{thm:structure} is to show that minimizing the function $f()$ of Lemma~\ref{lem:GoemansWilliamson} above is equivalent to minimizing a function that depends only on the differences between distinct efficiencies, and the $\Delta_{i,\ell}$ values. 

\begin{definition}
\label{def:g}
The function $g()$ defined for WSPT schedules $\sigma$ for $Pm||\sum w_j C_j$ instances with $k$ distinct efficiencies is given by    
$$
g(\sigma) =   
\sum_{i = 1}^{k}  \Big[ \big(e_{(i)}-e_{(i+1)}\big)  \cdot \sum^{m}_{\ell=1} \Delta^2_{i,\ell} \Big].
$$
\end{definition}

To prove that minimizing the function $g()$ of Definition~\ref{def:g} above is indeed equivalent to minimizing the function $f()$ of Lemma~\ref{lem:GoemansWilliamson}, we make use of the following auxiliary lemma. 

\begin{lemma} 
\label{lem:Tomohiro}
For any $m$ non-negative integers $x_1,\ldots,x_m \in \mathbb{N}$ we have
$$
\sum_{\ell = 1}^m x_\ell^2 \,\, = \,\, \sum_{\ell =1}^m (\mu - x_\ell)^2 +  m\mu^2,
$$ 
where $\mu = \frac{1}{m} \sum_{\ell=1}^m x_\ell$.
\end{lemma}
\begin{proof}
Expanding the term $\sum_{\ell = 1}^m (\mu - x_\ell)^2$ yields
\begin{align*}
\sum_{\ell = 1}^m (\mu - x_\ell)^2 = \sum_{\ell = 1}^m (\mu^2 - 2x_\ell \mu + x_\ell^2) = -m\mu^2 + \sum_{\ell = 1}^m x_\ell^2,
\end{align*}
and so the lemma holds. \qed
\end{proof}

\begin{lemma}
\label{lem:AlternativeForumlation}%
For any WSPT schedule $\sigma$ for $J$ we have
$$
2 f(\sigma) =  g(\sigma) + \sum^n_{j=1} w_j p_j +   \sum_{i = 1}^{k}\big(e_{(i)}-e_{(i+1)}\big)  \cdot m \mu^2_i.
$$
\end{lemma}

\begin{proof}
Consider some WSPT schedule $\sigma$. For each~$i \in \{1,\ldots,k\}$ and $\ell \in \{1,\ldots,m\}$, let $x_{i,\ell} = P(J^\ell_1 \cup \cdots \cup J^\ell_i)$. Then $\Delta_{i,\ell} = \mu_i - x_{i,\ell}$ for each $1 \leq \ell \leq m$, and so by Lemma~\ref{lem:Tomohiro} we have 
\begin{align*}
2 f(\sigma) \,\,&= \sum_{i = 1}^{k}  \Big[ \big(e_{(i)}-e_{(i+1)}\big)  \cdot \sum^m_{\ell=1} x_{i,\ell}^2 \Big] +  \sum^n_{j=1} w_j p_j \\ 
&= \sum_{i = 1}^{k}  \left[ \big(e_{(i)}-e_{(i+1)}\big)   \cdot  \left[ \sum^m_{\ell=1}\left(\mu_i-x_{i,\ell}\right)^2 +m \mu^2_i\right] \right] +  \sum^n_{j=1} w_j p_j\\
&= g(\sigma) + \sum^n_{j=1} w_j p_j +   \sum_{i = 1}^{k}\big(e_{(i)}-e_{(i+1)}\big)  \cdot m \mu^2_i , \\
\end{align*}
and the lemma thus follows. \qed
\end{proof}

Observe that $g(\sigma)$ is the only term in the right-hand side that depends on the schedule~$\sigma$; every other term depends only on the instance $J$ of $Pm|| \sum w_jC_j$. It follows that finding a schedule $\sigma$ that minimizes $f(\sigma)$ is equivalent to finding a schedule $\sigma$ that minimizes $g(\sigma)$.

\subsection{An upper bound on $g()$}

The next step is to show an upper bound on $g()$ for optimal schedules of $J$ that depends only on $w_{\max}$, $p_{\max}$, and $m$, and not on $n=|J|$.

\begin{lemma} 
\label{lemma:opt-bound}
There is a schedule $\sigma$ for $J$ with $g(\sigma) \leq m \cdot w_{\max} \cdot p_{\max}^2$.
\end{lemma}
\begin{proof}
We first show that there exists a WSPT schedule $\sigma$ with $|\Delta_{i,\ell}|  \leq p_{\max}$ for all $1 \leq i \leq k$ and $1 \leq \ell \leq m$. Consider a scheduling algorithm that first orders the jobs in non-increasing order of efficiencies. Then, it assigns the jobs to the machines, based on the above ordered, such that at any iteration $j \in \{1,\ldots,n\}$ job $j$ is assigned to the least loaded machine. This algorithm is usually refereed to as a load balancing algorithm, and is commonly used in the literature to obtain an approximate solution for scheduling problems on parallel machines (see, \emph{e.g.,}~\cite{Azar,Taub}). We claim that by applying this algorithm the difference between the load of any two machines in any of the iterations is at most~$p_{\max}$. 

This trivially holds for iteration $j=1$ since only one machine receives a job. Now, consider iteration~$j \geq 2$ and assume, without loss of generality, that the machine loads at the end of iteration~$j-1$ satisfy $L_1\leq L_2\leq\ldots \leq L_m$ with~$L_m-L_1 \leq p_{\max}$. In iteration $j$, job~$j$ is assigned to the least loaded machine, increasing its load to $L_1+p_j$, while other machine loads remain unchanged. Thus, for any machine $\ell \in \{2, \ldots, m\}$, we have:
$$
L_1+p_j-L_\ell \leq L_1+p_j-L_1=p_j \leq p_{\max},
$$
and
$$
L_\ell- (L_1 + p_j) \leq L_m- L_1 - p_j \leq p_{\max} - p_j \leq p_{\max}.
$$
Since the load of other machines remains unchanged, a similar bound holds for any pair of machines. Hence, at every iteration~$j$, the difference between the most and the least loaded machine remains at most~$p_{\max}$.

Consider now the iteration in which the last job that belongs to $J_i$ is assigned to the least loaded machine. As the average load on each machine~$\mu_i$ is both at most the maximal load and at least the minimal load, we have $|\Delta_{i,\ell}|  \leq p_{\max}$ for all $1 \leq i \leq k$ and $1 \leq \ell \leq m$. Now, as $e_{(1)} \leq w_{\max}$, we also have 
\begin{align*}
g(\sigma) = &\sum_{i = 1}^{k}  \Big[ \big(e_{(i)}-e_{(i+1)}\big)  \cdot \sum^{m}_{\ell=1} \Delta^2_{i,\ell} \Big] \le \sum_{i = 1}^{k}  \Big[ \big(e_{(i)}-e_{(i+1)}\big)  \cdot \sum^{m}_{\ell=1}  p^2_{\max} \Big] =\\ 
&m \cdot  p_{\max}^2 \cdot \sum_{i = 1}^{k} \big(e_{(i)}-e_{(i+1)}\big) = m \cdot  p_{\max}^2 \cdot e_{(1)} \le m \cdot w_{\max} \cdot p_{\max}^2,
\end{align*}
and so the lemma follows. \qed
\end{proof}

\subsection{Proof of Theorem~\ref{thm:structure}}

The final ingredient that we need before proving Theorem~\ref{thm:structure}, is the following easy lower bound on the difference between any two consecutive efficiencies:
\begin{lemma}
\label{lem:EfficienciesDiffernce}%
$e_{(i)}-e_{(i+1)} \geq 1/p_{\max}^2$ for all $i \in \{1,\ldots,k\}$.
\end{lemma}

\begin{proof}
Let $e_{(i)} = w/p$ and $e_{(i+1)} = w'/p'$ for some $w, w' \in \{1,\ldots,w_{\max}\}$ and $p, p' \in \{1,\ldots,p_{\max}\}$. Recall that $e_{(i)} > e_{(i+1)}$ which implies that $wp' > w'p$. Then 
\begin{align*}
e_{(i)}-e_{(i+1)} =\frac{w}{p} - \frac{w'}{p'} = \frac{wp' - w'p}{pp'} \ge \frac{1}{pp'} \ge \frac{1}{p_{\max}^2}, 
\end{align*}
and the lemma follows. \qed
\end{proof}


\begin{proof}[of Theorem~\ref{thm:structure}]
Let $\sigma^*$ be some WSPT schedule for $J$ where we have $\Delta_{i^*,\ell^*} >  \sqrt{m\cdot w_{\max}} \cdot p_{\max}^2$ for some $1 \leq i^* \leq k$ and $1 \leq \ell^* \leq m$. Then by Lemma~\ref{lem:EfficienciesDiffernce}, we get that
\begin{align*}
g(\sigma) &= \sum_{i = 1}^{k}  \Big[ \big(e_{(i)}-e_{(i+1)}\big)  \cdot \sum^{m}_{\ell=1} \Delta^2_{i,\ell} \Big] \geq \big(e_{(i^*)}-e_{(i^*+1)}\big) \cdot \Delta^2_{i^*,\ell^*} \\ 
& >   \big(e_{(i^*)}-e_{(i^*+1)}\big) \cdot m \cdot w_{\max} \cdot p^4_{\max} \geq  m\cdot w_{\max} \cdot p^2_{\max}.
\end{align*}
By Lemma~\ref{lemma:opt-bound} there exists a schedule $\sigma$ with $g(\sigma) \leq m \cdot w_{\max} \cdot p^2_{\max} < g(\sigma^*)$. This implies by Lemma~\ref{lem:AlternativeForumlation} that $f(\sigma) < f(\sigma^*)$, and so $\sigma^*$ is not optimal by Lemma~\ref{lem:GoemansWilliamson}. \qed
\end{proof}

\section{Dynamic Programming Algorithm}

We next apply our structural result stated in Theorem~\ref{thm:structure} to obtain a pseudo-polynomial time algorithm for $Pm|| \sum w_jC_j$. According to Theorem~\ref{thm:structure}, we can restrict our attention to WSPT schedules where for each $i \in \{1,\ldots,k\}$, the processing time of jobs with efficiency at least $e_{(i)}$ is distributed evenly on any pair of machines up to an additive factor of~$c = \lfloor\,\, (m \cdot w_{\max})^{1/2} \cdot p_{\max}^2 \,\,\rfloor = O(\sqrt{w_{\max}} \cdot p^2_{\max})$. This allows us to design an efficient dynamic programming algorithm in cases where $c$ is relatively small.
\main*
\noindent Our dynamic program will be over vectors $x \in [-c, c]^m$ whose entries will correspond to deviation in loads on each machine. Before proceeding to describe our dynamic program, we describe a few operations on vectors that we will use in our algorithm.


\subsection{Minkowski sums of vector sets}

We will be using vectors of varying dimension $d = O(1)$ whose entries correspond to processing times of various job sets. We consider component-wise addition between two such vectors $x,y \in \mathbb{N}^d$, with $x=(x_1,\ldots,x_d)$ and $y=(y_1,\ldots,y_d)$, and so $x+y=(x_1+y_1,\ldots,x_d+y_d)$. Let $X,Y \subseteq \mathbb{N}^d$. We define $X+Y$ as the \emph{Minkowski sum} of $X$ and $Y$; that is, $X + Y = \{x+y : x \in X, y \in Y\}$. It is well known that one can efficiently compute the Minkowski sum $X+Y$ using multivariate polynomial multiplication. We briefly review how this can be done below. 

Let $p_X[\alpha_1,\ldots,\alpha_d] = \sum_{(x_1,\ldots,x_d) \in X} \Pi_{i=1}^d \alpha_i^{x_i}$ and $p_Y[\beta_1,\ldots,\beta_d] = \sum_{(y_1,\ldots,y_d) \in Y} \Pi_{i=1}^m \beta_i^{y_i}$. Then the exponents of all terms in $p_X \cdot p_Y$ with non-zero coefficients correspond to elements in the Minkowski sum $X + Y$. Note that if $P$ is a bound on both $\max_i x_i$ and $\max_i y_i$, then the maximum degree for each variable of these polynomials is $P$. Multiplying two $d$-variate polynomials of maximum degree $P$ on each variable can be reduced to multiplying two univariate polynomials of maximum degree $O(P^d)$ using Kronecker's map (see \emph{e.g.}~\cite{Pan94}). Since multiplying two univariate polynomials of maximum degree $O(P^d)$ can be done in $O(P^d \log P)$ time~\cite{CormenLRC2009}, we get that $X+Y$ can be computed in $\widetilde{O}(P^d)$ time.

\begin{lemma}
\label{lem:MinkowskiSum}%
Let  $X,Y \subseteq \{0,1,\ldots,P\}^d$ for some $d \in \mathbb{N}$. Then the Minkowski sum $X+Y$ can be computed in~$\widetilde{O}(P^d)$ time.     
\end{lemma}

\subsection{Load distribution vectors}

For a subset of jobs $J' \subseteq J$ and an non-negative integer $d \leq m$, define $X^d(J')$ as the set of all $d$-dimensional vectors $(P(J'_1),\ldots,P(J'_d))$,  where $J'_1,\ldots,J'_d$ are pairwise disjoint subsets (not necessarily a partition) of $J'$. In this way, each vector in $X^d(J')$ represents a possible load distribution of some of the jobs in~$J'$ across the first $d$ machines. We can use Lemma~\ref{lem:MinkowskiSum} above to compute $X^d(J')$ rather efficiently: First, we split $J'$ into two sets $J'_1$ and $J'_2$ of roughly equal size. Then we recursively compute $X^d(J'_1)$ and $X^d(J'_2)$. Finally, we compute $X^d(J')$ by computing the Minkowski sum $X^d(J')= X^d(J'_1) + X^d(J'_2)$. As $X^d(\{j\})$ can be trivially computed in $O(d)=O(1)$ time for any job $j\in J$, the entire algorithm runs in~$\widetilde{O}(P(J')^d)$ time. 

Recall that $J_i = \{j \in J : w_j/p_j = e_{(i)}\}$ denotes the set of all jobs with efficiency $e_{(i)}$. For $i \in \{1,\ldots,k\}$, let $X_i$ denote the set of load distribution vectors of jobs in $J_i$. That is, $X_i$ is the set of all $m$-dimensional vectors $(P(J^1_i),\ldots,P(J^m_i))$, where $J^1_i,\ldots,J^m_i$ are pairwise disjoint and $J^1_i \cup \cdots \cup J^m_i = J_i$. We use the above observation to show that all sets $X_1,\ldots,X_k$ can be computed in $\tilde{O}(P^{m-1})$ time.

\begin{lemma}
\label{lem:LoadDistVectors}%
There is an algorithm that computes $X_1,\ldots,X_k$ in $\tilde{O}(P^{m-1})$ time.
\end{lemma}
\begin{proof}
The algorithm first computes all sets $X^{m-1}(J_1),\ldots,X^{m-1}(J_k)$ using the algorithm described above. This can be done in time
$$
\widetilde{O}(P(J_1)^{m-1} + \cdots + P(J_k)^{m-1}) = \widetilde{O}((P(J_1) + \cdots + P(J_k))^{m-1}) = \tilde{O}(P^{m-1}).
$$
To compute $X_i$ for each $i \in \{1,\ldots,k\}$, it then by appends to each vector $x \in X^{m-1}(J_i)$ an additional component~$x_m = P(J_i)-(x_1+\cdots +x_{m-1})$. Altogether this requires an additional $\tilde{O}(P^{m-1})$ time, and so the lemma follows. \qed
\end{proof}

\subsection{Load deviation vectors}

Next, for each $i \in \{1,\ldots,k\}$, we compute a set of $m$-dimensional vectors as follows. For each~$x\in X_i$ we construct a vector $y=y(x)$ where $y_\ell = P(J_i)/m - x_\ell$ for each~$\ell \in \{1.\ldots,m\}$. Let $Y_i=\{y(x): x \in X_i\}$ be the set of all vectors created this way.

We say that a vector $\Delta_i = (\Delta_{i,1},\ldots,\Delta_{i,m})$ is a \emph{load deviation vector} for $J_1 \cup \cdots \cup J_i$ if there exists a WSPT schedule $\sigma$ for $J$ with $\Delta_{i,\ell} = \mu_i-P(J_1^\ell \cup \dots \cup J_i^\ell)$ for all $1 \leq \ell \leq m$. We can use the vectors in $Y_i$ to compute all load deviation vectors $\Delta_i$.

\begin{lemma}
\label{lem:LoadDevVectors}%
Let $i \in \{1,\ldots,k\}$. An $m$-dimensional vector $\Delta_i$ is a load deviation vector for~$J_1 \cup \cdots \cup J_i$ iff $\Delta_i - y= \Delta_{i-1}$ for some $y \in Y_i$ and some load deviation vector~$\Delta_{i-1}$ for $J_1 \cup \cdots \cup J_{i-1}$ (where~$\Delta_0 = 0$). 
\end{lemma}

\begin{proof}
Suppose $\Delta_i$ is a load deviation vector for $J_1 \cup \cdots \cup J_i$, and let $\sigma$ be a WSPT schedule for $J$, where $\Delta_{i,\ell} = \mu_i-P(J_1^\ell \cup \cdots \cup J_i^\ell)$ for all $1 \leq \ell \leq m$. Let $\Delta_{i-1}$ be the $m$-dimensional vector defined by~$\Delta_{i-1,\ell} = \mu_{i-1}-P(J_1^\ell \cup \cdots \cup J_{i-1}^\ell)$ for all $1 \leq \ell \leq m$. Then $\Delta_{i-1}$ is a load deviation vector for~$J_1 \cup \cdots \cup J_{i-1}$. Now, consider the vector $x = (P(J^1_i),\ldots,P(J^m_i)) \in X_i$, and let $y=y(x) \in Y_i$. Then for all $1 \leq \ell \leq m$, we have
\begin{align*}
\Delta_{i,\ell} &= \mu_i-P(J_1^\ell \cup \cdots \cup J_i^\ell) \\ 
&= \frac{P(J_1 \cup \cdots \cup J_i)}{m} -  P(J_1^\ell \cup \dots \cup J_i^\ell) \\
&=  \frac{P(J_1 \cup \cdots \cup J_{i-1})}{m} + \frac{P(J_i)}{m} -  P(J_1^\ell \cup \cdots \cup J_{i-1}^\ell)-P(J^\ell_i) \\
&= 
y_\ell + \Delta_{i-1,\ell},
\end{align*}
and so the forward direction of the lemma holds. The backward direction holds as well by similar arguments. \qed    
\end{proof}


\subsection{Dynamic program}

Our dynamic program will be over $m$-dimensional load deviation vectors. For each $i \in \{0,\ldots,k\}$, we construct a dynamic programming table $T_i \colon [-c, c]^m \to \mathbb{N}$, where~$c=O(\sqrt{w_{\max}} \cdot p^2_{\max})$) as mentioned above. For each $\Delta_i \in [-c, c]^m$, entry $T_i[\Delta_i]$ stores the minimum value of $g(\sigma)$ among all WSPT schedules $\sigma$ for jobs $J_1\cup \dots \cup J_i$ such that $\Delta_{i,l}=\mu_i- P(J^\ell_1 \cup \cdots \cup J^\ell_i)$ for each $1 \leq i \leq k$ and $1 \leq \ell \leq m$. We define~$T_i[\Delta_i]=\infty$ if no such schedule exists. More formally, 
$$
T_i[\Delta_i] = \min_{\sigma} \sum_{j = 1}^{i}  \Big[ \big(e_{(j)}-e_{(j+1)}\big)  \cdot \sum^m_{\ell=1} \Delta^2_{j,\ell} \Big],
$$ 
where the minimum is taken over all WSPT schedules $\sigma$ where $\Delta_i$ is the load deviation vector of $\sigma$ for $J_1 \cup \cdots \cup J_i$. In this way, the entry with the minimum value in $T_k$ will correspond to the minimum value of $g(\sigma)$ for any WPST schedule~$\sigma$ for~$J$, which according to Lemma~\ref{lem:AlternativeForumlation} and Lemma~\ref{lem:GoemansWilliamson} is equivalent to the minimum total weighted completion time of the input set of jobs~$J$.

We compute each table $T_i$ as follows. For $i = 0$, we have $T_0[0] = 0$ and $T_0[\Delta_0] = \infty$ for all other values $\Delta_0 \in [-c, c]^{m} \setminus \{0\}$. Then, for $i \ge 1$, we compute $T_i[\Delta_i]$ for any $\Delta_i \in [-c,c]^m$ using the following recurrence:
\begin{align*}
T_i[\Delta_i] = \min_{y_i \in Y^*_i} \Big\{T_{i-1}[\Delta_i-y_i] \,\,+\,\, (e_{(i)}-e_{(i+1)}) \cdot \sum^m_{\ell=1}   \Delta_{i,\ell}^2\Big\},
\end{align*}
where~$Y^*_i = Y_i\cap [-2c, 2c]^{m}$ for each $i \in \{1,\ldots,k\}$. Naturally, if $\Delta_i-y_i \notin [-c, c]^{m}$, we set $T_{i-1}[\Delta_i-y_i]=\infty$ in the recursion above. Correctness of this recurrence follows directly from Lemma~\ref{lem:LoadDevVectors}, and from the fact that we can restricted ourselves to deviation vectors in $[-c,c]^m$ due to Theorem~\ref{thm:structure}. 

\subsection{Time complexity}
According to Lemma~\ref{lem:LoadDistVectors}, we can compute all sets of vectors $X_1,\ldots,X_k$ in $\widetilde{O}(P^{m-1})$ time. It follows that also the sets $Y^*_1,\ldots,Y^*_k$ can be computed in similar time. Note that entries~$y_{i,\ell}$ in each vector $y_i \in Y^*_i$ are fractional, but may only take up to $4cm$ different values, since~$m \cdot y_{i,\ell}$ is integer and $-2cm \leq m \cdot y_{i,\ell} \leq 2cm$. Thus, we have $|Y^*_i| \le (4cm)^m$, and each entry~$T_i[\Delta_i]$ can be computed from $T_{i-1}$ in $O(c^m)$ time. For a similar reason, each entry $\Delta_{i,\ell}$ in a load deviation vector $\Delta_{i}$ for $J_1 \cup \ldots \cup J_i$ may be fractional, but can only take up to $2cm$ different values. Thus, there are at most $(2cm)^m$ different entries per table, and each table can be computed in $O(c^{2m})$ time. It follows that  computing all tables $T_0, \dots, T_{k}$  requires $O(k c^{m})$ time, which altogether gives us a running time of 
$$
\widetilde{O}(P^{m-1}) +  O( k \cdot ( \sqrt{w_{\max}} \cdot p_{\max}^2)^{2m}) = \widetilde{O}(P^{m-1} + w_{\max}^{m+1} \cdot p_{\max}^{4m+1}),
$$
where the last equality follows from the fact that~$k$ is upper bounded by  $w_{\max} \cdot p_{\max}$. This completes the proof of 
Theorem~\ref{thm:main}.

\section{The Two Machine Case}

In the following section we present an improved algorithm for the two machine case, \emph{i.e.} the $P2|| \sum w_j C_j$ problem. This algorithm closely resembles the one described above, but relies on a slightly different structure theorem. Rather than focusing on the deviation of each machine's load from the average, it is based on the load difference between the two machines.

For a given WSPT schedule $\sigma$ and $i \in \{1,\ldots,k\}$, let $\delta_i$ denote  the difference in the load of jobs with efficiency at least~$e_{(i)}$ between the first and second machine; that is, $\delta_i = P(J^1_1 \cup \cdots J^1_i) - P(J^2_1 \cup \cdots J^2_i)$. We define an alternative version of the function $g()$ in Definition~\ref{def:g} using the $\delta_i$ values. 
\begin{definition}
\label{def:g2}
The function $g_2()$ defined for WSPT schedules $\sigma$ for $Pm||\sum w_j C_j$ instances with $k$ distinct efficiencies is given by    
$$
g_2(\sigma) =   
\sum_{i = 1}^{k}  \big(e_{(i)}-e_{(i+1)}\big)  \cdot \delta^2_i.
$$
\end{definition}

Again, finding a WSPT schedule minimizes $g_2()$ is equivalent to finding a WSPT schedule minimizes $f()$ (and hence also the total weighted completion time), as can be seen by the following lemma.
\begin{lemma}
\label{lem:AlternativeForumlation2}%
For any WSPT schedule $\sigma$ for $J$ we have
$$
2 \cdot f(\sigma) = \frac{1}{2} \cdot g_2(\sigma) + \frac{1}{2} \cdot \sum_{i = 1}^{k}  \Big[ \big(e_{(i)}-e_{(i+1)}\big)  \cdot P(J_1 \cup \cdots \cup J_i)^2 \Big] + \sum^n_{j=1} w_j p_j.
$$
\end{lemma}

\begin{proof}
Consider some WSPT schedule $\sigma$ for $J$. For each~$i \in \{1,\ldots,k\}$ and $\ell \in \{1,2\}$, let $x_{i,\ell} = P(J^\ell_1 \cup \cdots \cup J^\ell_i)$. Then $\delta_i = x_{i,1} - x_{i,2}$ for each~$i \in \{1,\ldots,k\}$, and so by Lemma~\ref{lem:GoemansWilliamson} we have
\begin{align*}
2 f(\sigma) \,\,&= \sum_{i = 1}^{k}  \Big[ \big(e_{(i)}-e_{(i+1)}\big)  \cdot (x^2_{i,1}+x^2_{i,2}) \Big] +  \sum^n_{j=1} w_j p_j \\ 
&= \sum_{i = 1}^{k}  \left[ \big(e_{(i)}-e_{(i+1)}\big)  \cdot \frac{1}{2} \cdot  \Big[ (x_{1,1}  + x_{i,2})^2 + (x_{1,1}  - x_{i,2})^2  \Big] \right] +  \sum^n_{j=1} w_j p_j\\
&= \sum_{i = 1}^{k}  \left[ \big(e_{(i)}-e_{(i+1)}\big)  \cdot \frac{1}{2} \cdot  \Big[ P(J_1 \cup \cdots \cup J_i)^2 + \delta^2_i  \Big] \right] +  \sum^n_{j=1} w_j p_j\\
&= \frac{1}{2} \cdot g_2(\sigma) + \frac{1}{2} \cdot \sum_{i = 1}^{k}  \Big[ \big(e_{(i)}-e_{(i+1)}\big)  \cdot P(J_1 \cup \cdots \cup J_i)^2 \Big] + \sum^n_{j=1} w_j p_j. \\ 
\end{align*}
The lemma thus follows. \qed
\end{proof}

We can now use the same arguments used for proving Theorem~\ref{thm:structure} to show that the $\delta_i$'s are bounded in any optimal WSPT schedule. Indeed, the WSPT schedule $\sigma$ used in the proof of Lemma~\ref{lemma:opt-bound} has $\delta_i \leq p_{\max}$ for each $i \in \{1,\ldots,k\}$, and so $g(\sigma) \leq w_{\max} \cdot p_{\max}^2$. Thus, since the difference between any two consecutive efficiencies $e_{(i)}-e_{(i+1)}$ is at least $p^2_{\max}$, we find that any WSPT schedule $\sigma$ with $\delta_i > \sqrt{w_{\max}} \cdot p^2_{\max}$ has $g(\sigma) > w_{\max} \cdot p_{\max}^2$, and therefore cannot be optimal.

\begin{theorem}
\label{thm:structure2}
In any optimal WSPT schedule for any $P2||\sum w_j C_j$ instance with $k$ distinct efficiencies we have 
$$
|\delta_i| \leq  \sqrt{w_{\max}} \cdot p_{\max}^2
$$ 
for each $1 \leq i \leq k$.    
\end{theorem}

Our dynamic program will be over all possible values of $\delta_i$'s in a WSPT schedule for~$J$. Let $c = \lfloor \sqrt{w_{\max}} \cdot p^2_{\max} \rfloor$. For each $i \in \{0,\ldots,k\}$, we construct a dynamic programming table $T_i:[-c,c] \to \mathbb{N}$ where entry $T_i[x]$ stores the minimum value of
$$
T_i[\delta_i] = \min_{\sigma} \sum_{j = 1}^{i}   \big(e_{(j)}-e_{(j+1)}\big)  \cdot  \delta^2_j,
$$ 
where the minimum is taken over all WSPT schedules $\sigma$ in which $\delta_i$ is the difference in load between the first and second machine for jobs with efficiency at least~$e_{(i)}$. To compute all tables $T_i$, we again construct all load distribution vectors $X_1,\ldots,X_k$ using Lemma~\ref{lem:LoadDistVectors}. Then for each $i \in \{1,\ldots,k\}$, we construct the set $Y_i=\{x_1 - x_2 : x \in X_i\}$. It now holds for each $i \in \{1,\ldots,k\}$ that if $\delta_i$ is indeed the difference in the load of the two machines on jobs with efficiency at least~$e_{(i)}$ for some WSPT schedule, then $\delta_i - y= \delta_{i-1}$ for some $y \in Y_i$, and vice-versa. This yields the following recurrence:
\begin{align*}
T_i[\delta_i] = \min_{y \in Y^*_i} \Big\{T_{i-1}[\delta_i-y] + (e_{(i)}-e_{(i+1)}) \cdot x^2\Big\},
\end{align*}
where~$Y^*_i = Y_i \cap [-2c, 2c]$ for each $i \in \{1,\ldots,k\}$. Again, for $\delta_i-y \notin [-c, c]$ we let $T_{i-1}[\delta_i-y]=\infty$ in the recursion above.

Regarding the time complexity of this algorithm, first observe that all sets $X_i$, $Y_i$, and~$Y^*_i$, for $1 \leq i \leq k$, can be computed in $\widetilde{O}(P)$ time. As $|Y^*_i| = O(\sqrt{w_{\max}} \cdot p^2_{\max})$, each entry $T_i[\delta_i]$ can be computed from $T_{i-1}$ in $O(\sqrt{w_{\max}} \cdot p^2_{\max})$ time. As there are $O(\sqrt{w_{\max}} \cdot p^2_{\max})$ different entries per table, each table can be computed in $O(w_{\max} \cdot p^4_{\max})$ time, and computing all tables $T_0, \dots, T_{k}$ requires $O(k \cdot w_{\max} \cdot p^4_{\max})$ time. Altogether, as~$k$ is upper bounded by $w_{\max} \cdot p_{\max}$, this gives us a running time of 
$$
\widetilde{O}(P) + O(k \cdot w_{\max} \cdot p^4_{\max}) = \widetilde{O}(P+ w_{\max}^{2} \cdot p_{\max}^{5}).
$$ 
Theorem~\ref{thm:secondary} thus follows.

\section{Discussion}

In this work, we have presented a pseudo-polynomial time algorithm for minimizing the total weighted completion time on parallel machines, which improves upon the classical Lawler and Moore algorithm for instances with bounded processing times and weights.
A key element of our approach is a new structural theorem about optimal schedules. In particular, we show that when jobs are grouped according to their efficiencies, the deviation in load distribution is bounded by a polynomial in $w_{\max}$ and $p_{\max}$ (see Theorem~\ref{thm:structure}). This bound is instrumental in the design of our dynamic programming algorithm.

Looking ahead, an intriguing open question is whether this bound can be improved. Specifically, can we obtain a bound that depends on only one of $w_{\max}$ or $p_{\max}$? Furthermore, it remains open whether there exists a nearly linear-time algorithm when only the weights or only the processing times are bounded. Finally, it is also interesting to explore the case where the number of machines~$m$ is unbounded.

\section*{Acknowledgements}
Tomohiro Koana is supported by the European Research Council (ERC) under the European Union’s
Horizon 2020 Research and Innovation Programme (project CRACKNP, Grant Number 853234) and JSPS KAKENHI (Grant Number JP20H05967).

\bibliographystyle{plain}
\bibliography{biblo}

\end{document}